\documentclass[10pt]{article}
\usepackage{amsmath}   
\usepackage{amssymb}   
\usepackage{amsthm}    
\usepackage{stmaryrd}  
\usepackage{titletoc}  
\usepackage{mathrsfs}  
\usepackage{enumitem}
\usepackage{color}

\theoremstyle{plain}
\newtheorem{thm}{Theorem}[section]
\newtheorem{cor}[thm]{Corollary}
\newtheorem{lem}[thm]{Lemma}
\newtheorem{prop}[thm]{Proposition}

\theoremstyle{definition}
\newtheorem{defn}{Definition}[section]
\newtheorem{ass}{Assumption}[section]
\newtheorem{rmk}{Remark}[section]

\DeclareMathOperator*{\esssup}{ess\,sup}
\DeclareMathOperator*{\essinf}{ess\,inf}

\newcommand{\cL}{\mathcal{L}}

\newcommand{\cB}{\mathcal{B}}

\newcommand{\bP}{\mathbb{P}}
\newcommand{\bR}{\mathbb{R}}
\newcommand{\bN}{\mathbb{N}}

\newcommand{\sF}{\mathscr{F}}
\newcommand{\sP}{\mathscr{P}}

\renewcommand{\baselinestretch}{1.1}

\makeatletter\@addtoreset{equation}{section} \makeatother
 \allowdisplaybreaks

\begin{document}

\title{A Constrained Control Problem with Degenerate Coefficients and Degenerate Backward SPDEs with Singular Terminal Condition\thanks{We {thank} seminar participants at various institutions for valuable comments and suggestions. UH acknowledges support through the SFB 649 ``Economic Risk''; {QZ is supported by NSF of China (No. 11101090, 11471079) {and the Science and Technology Commission of Shanghai Municipality (No. 14XD1400400). The paper was finished while UH was visiting the Center for Interdisciplinary Research (ZiF), Bielefeld University. Grateful acknowledgement is made for hospitality.}}}}

\author{ Ulrich Horst$^{1}$ \and Jinniao Qiu$^{2}$ \and Qi Zhang$^{3}$}

%
\footnotetext[1]{Department of Mathematics and School of Business and Economics, Humboldt-University of Berlin, Unter den Linden~6, 10099 Berlin, Germany. \textit{E-mail}: \texttt{horst@math.hu-berlin.de}}
\footnotetext[2]{Department of Mathematics, Humboldt-University of Berlin, Unter den Linden~6, 10099 Berlin, Germany. \textit{E-mail}:  \texttt{qiujinn@gmail.com}}
\footnotetext[3]{School of Mathematical Sciences, Fudan University,
            Shanghai 200433, China.
            \textit{E-mail}: \texttt{qzh@fudan.edu.cn}}

\maketitle

\begin{abstract}
  We study a constrained optimal control problem with possibly degenerate coefficients arising in models of optimal portfolio liquidation under market impact. The coefficients can be random in which case the value function is described by a degenerate backward stochastic partial differential equation (BSPDE) with singular terminal condition. For this degenerate BSPDE, we prove existence and uniqueness of a nonnegative solution. Our existence result requires a novel gradient estimate for degenerate BSPDEs.  
\end{abstract}

\bigskip

{\bf AMS Subject Classification:} 93E20, 60H15, 91G80

{\bf Keywords:} stochastic control, degenerate parabolic condition, backward stochastic partial differential equation, portfolio liquidation, singular terminal value.

\renewcommand{\baselinestretch}{1.1}\normalsize

\section{Introduction}

Let $T\in(0,\infty)$ and  $(\Omega,\bar{\sF},\bP)$ be a probability space equipped with a filtration $\{\bar{\sF}_t\}_{0 \leq t \leq T}$ which satisfies the usual conditions.  The probability space carries an $m$-dimensional Brownian motion $W$ and an independent point process $\tilde{J}$ on a non-empty Borel set $\mathcal{Z}\subset\bR^l$ with characteristic measure $\mu(dz)$. We endow the set $\mathcal{Z}$ with its Borel $\sigma$-algebra $\mathscr{Z}$ and denote by $\pi(dt,dz)$ the associated Poisson random measure. The filtration generated by $W$, together with all $\bP$ null sets, is denoted by $\{\sF_t\}_{t\geq 0}$. The predictable $\sigma$-algebra on $\Omega\times[0,+\infty)$ corresponding to $\{\sF_t\}_{t\geq0}$ and $\{\bar{\sF}_t\}_{t\geq0}$ is denoted $\sP$, respectively, $\bar{\sP}$. 

In this paper we address the following stochastic optimal control problem with constraints:
\begin{align}
  \min_{\xi,\rho}
  E\left[\int_0^T\!\!\!\Big(\eta_s(y_s)|\xi_s|^2+\lambda_s(y_s)|x_s|^2  \Big)\,ds
  +\!\int_0^T\!\int_{\mathcal {Z}} \!\!\!\gamma_s(y_s,z)|\rho_s(z)|^2\,\mu(dz)ds \right]
  \label{min-contrl-probm}
\end{align}
subject to
\begin{equation}\label{state-proces-contrl}
\left\{
\begin{split}
&x_t=x-\!\int_0^t\!\!\!\xi_s\,ds-\!\int_0^t\!\int_{\mathcal {\mathcal {Z}}}\rho_s(z)\,\pi(dz,ds),\,\,
\,t\in[0,T] \\
& x_T=0\\
& y_t=y+\!\int_0^t\!\!\!b_s(y_s)\,ds
  +\!\int_0^t\!\!\!\sigma_s(y_s)\,dW_s.
\end{split}
\right.
\end{equation}
The real-valued process $(x_t)_{t\in[0,T]}$ is the {\sl state process}. It is governed by a pair of {\sl controls} $(\xi,\rho)$. The $d$-dimensional process $(y_t)_{t\in[0,T]}$ is  uncontrolled. We sometimes write $x^{s,x,\xi,\rho}_t$ for $0\leq s\leq t\leq T$ to indicate the dependence of the state process on the control $(\xi,\rho)$, the initial time $s$ and initial state $x\in \mathbb{R}$. Likewise, we sometimes write $y^{s,y}_t$ to indicate the dependence on the initial time and state. 

The set of {\sl admissible controls} consists of all pairs $(\xi,\rho)\in\cL^2_{\bar{\sF}}(0,T)\times \cL^{2}_{\bar{\sF}}(0,T;L^2(\mathcal{Z}))$ s.t.  $x_{T}=0$ {a.s.} The {\sl cost functional} is assumed to be of the quadratic form:
\begin{equation}\label{cost0}
  J_t(x_t,y_t;\xi,\rho)=
  E\left[\int_t^T\!\!\!\Big(\eta_s(y_s)|\xi_s|^2+\lambda_s(y_s)|x_s|^2  \Big)\,ds
  +\!\int_t^T\!\int_{\mathcal {Z}} \!\!\!\gamma_s(y_s,z)|\rho_s(z)|^2\,\mu(dz)ds\Big|\sF_t  \right].
\end{equation}
The {\sl value function} is given by:
\begin{equation}\label{value-func}
  V_t(x,y){\triangleq}\essinf_{\xi,\rho} J_t(x_t,y_t;\xi,\rho)\big|_{x_t=x,y_t=y}.
\end{equation}

Control problems of the above form arise in models of optimal portfolio liquidation. In such models $x_t$ denotes the portfolio an investor holds at time $t \in [0,T]$, $\xi_t$ is the rate at which the stock is purchased or sold in a regular exchange at that time, $x_T=0$ is the {\sl liquidation constraint}, $\rho_t$ describes the number of stocks placed in a crossing network, $\pi$ governs the order execution in the crossing network and $y_t$ is a stochastic factor that drives the cost of liquidation. We refer to  \cite{AlmgrenChriss00,AnkirchnerJeanblancKruse13, HorstNaujokat13,KratzSchoeneborn13} and references therein for a detailed discussion of portfolio liquidation problems and an interpretation of the coefficients $\eta, \lambda$ and $\gamma$.

In a Markovian framework where all coefficients are deterministic functions of the control and state variables, the Hamilton-Jacobi-Bellman (HJB) equation turns out to be a \emph{deterministic} nonlinear parabolic partial differential equation (PDE) with a singularity at the terminal time; see \cite{GraeweHorstSere13} for details. Non-Markovian control problems with pre-specified terminal values have been studied in recent papers by Ankirchner, Jeanblanc {and Kruse \cite{AnkirchnerJeanblancKruse13},} and Graewe, Horst and Qiu \cite{GraeweHorstQui13}. Ankirchner et al. represented the value function in terms of a nonlinear backward stochastic differential equation (BSDE). 
%
The BSPDE-approach in \cite{GraeweHorstQui13} is more general. There, the authors construct the optimal control in feedback form assuming that there exists another independent $n$-dimensional Brownian motion $B$ s.t.
\begin{align}\label{eq_intro_Y_GHQ}
 y_t=y+\!\int_0^t\!\!b_s(y_s)\,ds
  +\!\int_0^t\!\!\sigma_s(y_s)\,dW_s + \!\int_0^t\!\bar{\sigma}_s(y_s)\,dB_s,
\end{align}
where the coefficients $b$, $\sigma$, $\bar{\sigma}$, $\lambda$, $\gamma$ and $\eta$ are measurable with respect to the filtration $\sF$ generated by $W$, and $\bar{\sigma}$ satisfies the {\sl super-parabolicity condition:} 
   \begin{align*}
       \sum_{k=1}^m\sum_{i,j=1}^d\bar\sigma^{ik}\bar\sigma^{jk}(t,x)\xi^i\xi^j\geq \delta |\xi|^2\quad \text{a.s.,} \,\,\forall\, (t,x,\xi)\in [0,T]\times\bR^d\times\bR^d.
   \end{align*}
   
We do not require the super-parabolicity condition. In a portfolio liquidation framework this condition is in fact not always natural; it is neither satisfied for many diffusion models of asset prices, nor for important absolutely continuous factors driving liquidation costs such as volume weighted average prices. For a given stock price process $(S_t)$ and a given model of aggregate intraday trading activities $(q_t)$, often a deterministic convex function, VWAP is defined as $v_t \triangleq \frac{\int_0^t S_u q_u du}{\int_0^t q_u du}$. This calls for an extension of the existing literature beyond the super-parabolic framework.

The constrained optimal control problem \eqref{min-contrl-probm} can be formally written as an unconstrained one:
\begin{align}
  \min_{\xi,\rho}
  E\left[\int_0^T\!\!\!\Big(\eta_s(y_s)|\xi_s|^2+\lambda_s(y_s)|x_s|^2  \Big)\,ds
  +\!\int_0^T\!\int_{\mathcal {Z}} \!\!\!\gamma_s(y_s,z)|\rho_s(z)|^2\,\mu(dz)ds
  +(+\infty)|x_T|^21_{\{x_T\neq 0\}} \right]
  \label{min-contrl-probm-unconstrain}
\end{align}
subject to
\begin{equation*}
\left\{
\begin{split}
&x_t=x-\!\int_0^t\!\!\!\xi_s\,ds-\!\int_0^t\!\int_{\mathcal {\mathcal {Z}}}\rho_s(z)\,\pi(dz,ds),\,\,
\,t\in[0,T];\\
& y_t=y+\!\int_0^t\!\!\!b_s(y_s)\,ds
  +\!\int_0^t\!\!\!\sigma_s(y_s)\,dW_s.
\end{split}
\right.
\end{equation*}

In view of Peng's seminal work \cite{Peng_92} on non-Markovian stochastic optimal control and the linear-quadratic structure of the cost functional, the dynamic programming principle suggests that the value function is of the form
\[
	V_t(x,y) = u_t(y) x^2{,}
\]
where $u$ is the first component of the pair $(u,\psi)$ satisfying formally the following backward stochastic partial differential equation (BSPDE) with singular terminal condition:
\begin{equation}\label{BSPDE-singlr}
  \left\{\begin{array}{l}
  \begin{aligned}
  -du_t(y)=\, &\bigg[\textrm{tr}\left( \frac{1}{2}\sigma_t(y)\sigma^*_t(y) D^2 u_t(y)
        + D \psi_t(y)\sigma^{*}_t(y)\right)+b^{*}_t(y)D u_t(y)
        + F(s,y,u_t(y)) \Big]\, dt  \\
        &-\psi_t(y)\, dW_{t}, \quad (t,y)\in [0,T]\times \bR^d;\\
    u_T(y)=\,& +\infty,  \quad y\in\bR^d{.}
    \end{aligned}
  \end{array}\right.
\end{equation}
{Here}
\begin{equation} \label{F}
  F(t,y,r) {\triangleq} -
        \!\int_{\mathcal{Z}}\!\frac{r^2}{\gamma(t,y,z)+r}\mu(dz)
        -\frac{r^2}{\eta_t(y)}+\lambda_t(y), \quad
        (t,y,r)\in \bR_+\times \bR^d\times \bR.
\end{equation}

BSPDEs were first introduced by Bensoussan \cite{Bensousan_83} as the adjoint equations of forward SPDEs and have since been extensively used in the stochastic control literature including \cite{DuTangZhang-2013,Du_Zhang_DegSemilin2012,EnglezosKaratzas09,Haussman-1987,QiuWei-RBSPDE-2013}. BSDEs with singular terminal conditions were first studied by Popier \cite{Popier06}. {To the best of our knowledge, {\sl degenerate} BSPDEs with singular terminal conditions have never been studied before. }

{As in PDE theory, degenerate BSPDEs are fundamentally different from super-parabolic ones (see \cite{DuQiuTang10,DuTangZhang-2013,Hu_Ma_Yong02}).}  Using recent results on degenerate BSPDEs \cite{DuTangZhang-2013,Du_Zhang_DegSemilin2012,Hu_Ma_Yong02,ma2012non}, standard arguments show that the BSPDE \eqref{BSPDE-singlr} has a unique solution and satisfies a comparison principle if the terminal value is finite. To show that a solution $(u,\psi)$ to the BSPDE with singular terminal value can be obtained as the limit of a sequence of solutions to with finite terminal values requires a gradient estimate for $u$, besides a growth condition of the limit near the terminal time. In the non-degenerate case analyzed in \cite{GraeweHorstQui13} such gradient estimate is not needed. The non-degenerate case only requires the growth condition on $u_t$ near the terminal time in which case the super-parabolicity guarantees sufficient regularity of $u$. 

The gradient estimate for a solution to a \emph{degenerate} BSPDE generally depends on its gradient at terminal time. In our case, the terminal value of the BSPDE is singular and hence it does in no obvious way characterize the gradient. Instead, we derive our gradient estimate from the gradient estimates of the approximating sequence. Our estimate seems new even in the Markovian case. Along with the gradient estimate, an explicit asymptotic estimate for the solution of our BSPDE near the terminal time is given. 

Due to the degeneracy of the diffusion coefficient, no generalized It\^o-Kunita formula {(like the one used in \cite{GraeweHorstQui13})} for the random filed $u$ satisfying the BSPDE \eqref{BSPDE-singlr} in the distributional sense is available. To prove the verification theorem we appeal instead to the link between degenerate BSPDEs and forward-backward stochastic differential equations (FBSDEs), which requires certain regularity of the random field $u$ including the gradient estimate mentioned above. Finally, using the It\^o formula for the square norm of the positive part of the solutions for BSPDEs, we prove that the obtained solution is the unique nonnegative one.

{The remainder of this paper is organized as follows. In Section 2, we introduce auxiliary notation and state our main result. In Section 3, we show that BSPDE \eqref{BSPDE-singlr} has a non-negative solution. The verification theorem and the uniqueness of solution are proved in Section 4. Selected results on semi-linear degenerate BSPDEs are recalled in an appendix where {we also show the well-posedness of the truncated version of BSPDE  \eqref{BSPDE-singlr} and establish a comparison principle for degenerate BSPDEs allowing for singular terminal values.} 
%

\section{Preliminaries and Main Result}

\subsection{Notation}

Throughout this paper, we use the following notation. $D$ and $D^2$ {denote} the first {order} and second order derivative {operators}, respectively; partial derivatives are denoted by $\partial$. For a Banach space $U$ and real {number} $p\in[1,\infty)$, we denote by $\cL^\infty_{\bar{\sF}}(0,T;U)$ and $\cL^p_{\bar{\sF}}(0,T;U)$ the Banach {spaces} of {all $\bar{\sP}$-progressively measurable $U$-valued processes which are essentially bounded and $p$-th integrable, respectively}. The spaces $\cL^p_{{\sF}}(0,T;U)$, $p \in [1,\infty]$, are defined analogously with $\bar{\sP}$ replaced by ${\sP}$.
For $k \in \mathbb{N}^+$ and $p\in [1,\infty)$, $H^{k,p}$ is the Sobolev space of all real-valued functions $\phi$ whose up-to $k$th order derivatives belong to $L^p(\bR^d)$, equipped with the usual Sobolev norm $\|\phi\|_{H^{k,p}}$.
For $k=0$, $H^{0,p}\triangleq L^p(\bR^d)$. Moreover,
\[
	H_{loc}^{k,p}\triangleq \{u;\,u\psi\in H^{k,p},\,\forall\,\psi\in C_c^{\infty}(\bR^d)\}
\]	
with $C_c^{\infty}(\bR^d)$ being the set of all the infinitely differentiable functions with compact support in $\bR^d$,
and $\cL^p(0,T;H_{loc}^{k,p})$ is defined as usual. For simplicity, by  $u=(u_1,\ldots,u_l)\in H^{k,p}$, $l\in\mathbb{N}^+$, we mean $u_1,\ldots,u_l\in H^{k,p}$ and $\|u\|_{H^{k,p}}^p{\triangleq}\sum_{j=1}^l\|u_j\|_{H^{k,p}}^p$. We use $\langle\cdot,\,\cdot\rangle$ and $\|\cdot\|$ to denote the inner product and {the} norm {in the usual Hilbert space $L^2(\bR^d)$ ($L^2$ for short), respectively.} We denote by $C^w_{\sF}([0,T];H^{k,p})$ the space of all $H^{k,p}$-valued and jointly measurable processes $(X_{t})_{t\in [0,T]}$ which are $\mathscr{F}$-adapted, a.s. weakly continuous with respect to $t$ on $[0,T]$\footnote{This means that for any $f\in(H^{k,p})^*$, the dual space of $H^{k,p}$, the mapping $t \mapsto f(X_t)$ is a.s. continuous on $[0,T]$.} and
$${E} \left[\sup_{t\in[0,T]}\|X_t\|_{H^{k,p}}^p\right] < \infty.$$ In the sequel, we write for any positive integer $k$
$$S^w_{\sF}([0,T];H^{k,p}){\triangleq}C^w_{\sF}([0,T];H^{k,p})\cap L^2(\Omega,\sF;C([0,T];H^{k-1,p})),\quad p\in[1,\infty].$$

\subsection{Assumptions and main result}

We now define what we mean by a solution to a BSPDE whose terminal value may be infinite.

\begin{defn}\label{defn-solution}
Let $G:\,\Omega\rightarrow \bar\bR:=[-\infty,+\infty]$ be $\sF_T/\cB(\bar\bR)$-measurable. A pair of processes $(u,\psi)$ is a solution to the BSPDE
  \begin{equation*}
  \left\{\begin{array}{l}
  \begin{aligned}
  -du_t(y)=\, &f(t,y,u,Du,D^2u,\psi,D \psi)\, dt
  -\psi_t(y)\, dW_{t}, \quad (t,y)\in [0,T]\times \bR^d;\\
    u_T(y)=\,& G(y),  \quad y\in\bR^d
    \end{aligned}
  \end{array}\right.
\end{equation*}
if $(u,\psi)\in \cL^2_{\sF}(0,\tau;H^{1,2}_{loc}) \times \cL^2_{\sF}(0,\tau;H_{loc}^{0,2}) $ for any $\tau \in(0,T)$, $\lim_{\tau\rightarrow T-} u_{\tau}(y)=G(y)$ a.e. in $\bR^d$ a.s. 
and for any $\varphi\in C_c^{\infty}(\bR^d)$, $\left\langle \varphi,\,f(\cdot,\cdot,u,Du,D^2u,\psi,D \psi)\right\rangle \in \cL^2(0,\tau;\bR)$ and
  \begin{align*}\label{HQZ2}
    \langle \varphi,\, u_t\rangle
    &= \langle\varphi,\, u_\tau\rangle+\!
    \int_t^\tau\!\!\!\langle \varphi,\, f(s,y,u,Du,D^2u,\psi,D \psi)\rangle\,ds -\!\int_t^\tau\!\!\langle\varphi,\,\psi_sdW_s\rangle\, \, {\text{a.s.},\, \forall \, 0\leq t\leq \tau<T.}
  \end{align*}
\end{defn}


We establish existence of a solution to BSPDE \eqref{BSPDE-singlr} under the following regularity conditions on the random coefficients. 

\begin{ass}
\begin{itemize}
	\item [(H.1)]  The functions
$b,\sigma,\eta,\lambda: \Omega\times[0,T]\times\bR^d\longrightarrow \bR^d\times \bR^{d\times m}\times  \bR_+\times \bR_+$ are
$\sP\times \mathscr{B}(\bR^d)$-measurable and essentially bounded by $\Lambda>0$, $\gamma: \Omega\times[0,T]\times\bR^d\times\mathcal{Z}\longrightarrow [0,+\infty]$
is $\sP\times \mathscr{B}(\bR^d)\times \mathscr{Z}$-measurable. Moreover,
there exists a positive constant $\kappa$ s.t. a.s.
$$
\eta_s(y)\geq \kappa,\quad \forall\,(y,s)\in\bR^d\times[0,T].
$$
\item[(H.2)] The first-order derivatives of $b$, $\eta$, $\lambda$ and the up to second-order derivatives of $\sigma$ exist and are bounded by some $L>0$ uniformly for any $(\omega,t)\in \Omega\times [0,T]$.
\item[(H.3)] There exists $(T_0,p_0)\in[0,T)\times (2,\infty)$ s.t. $$\essinf_{(\omega,t,y)\in\Omega\times[T_0,T]\times \bR^d} \eta_t(y) \geq \bigg(1-\frac{1}{2p_0}\bigg)
  \esssup_{(\omega,t,y)\in\Omega\times[T_0,T]\times \bR^d} \eta_t(y).$$
\end{itemize}
\end{ass}

{ The first two conditions above are standard and adopted throughout. The third is particular to the degenerate case. It guarantees sufficient integrability of the derivative of the value function and is satisfied if, for instance, $\eta\in C([t_0,T];L^{\infty}(\Omega\times\bR^d))$ with $\eta_T(\cdot)$ being a positive constant for some $t_0\in [0,T)$.

In view of (H.1) the random variable $F(\cdot,\cdot,0)$ belongs to $\cL^{\infty}_{\sF}(0,T;L^{\infty}(\bR^d))$ where $F$ is defined in (\ref{F}). Since it is more convenient to work with a BSPDE whose driver belongs to $L^p$ \cite{DuQiuTang10,DuTangZhang-2013,Du_Zhang_DegSemilin2012} we use the weight function 
\begin{equation}
\label{def-weight-func}
\theta(y)=(1+|y|^2)^{-q}\quad \textrm{for} \quad  y\in\mathbb{R}^d,\quad \text{with }q>d
\end{equation}
so that $\theta F(\cdot,\cdot,0)\in \cL^{p}(0,T;H^{1,p})$ for any $p\in[1,\infty)$. As in \cite{GraeweHorstQui13} a direct computation shows that $(u,\psi)$ solves \eqref{BSPDE-singlr} if and only if $(v,\zeta){\triangleq}(\theta u,\theta\psi)$ solves the BSPDE \eqref{BSPDE-singlr-1} given in the appendix.} We are now ready to state our main results. The following is a summary of Theorems \ref{thm-existence}, \ref{thm-verification} and \ref{thm-minimal}.

\begin{thm}\label{thm-conclude}
Under Conditions (H.1)--(H.3) the BSPDE \eqref{BSPDE-singlr} admits a unique nonnegative solution $(u,\psi)$, i.e., for any solution $(\bar{u},\bar{\psi})$ to BSPDE \eqref{BSPDE-singlr} satisfying
  $$
(\theta\bar{u},\theta\bar{\psi}+\sigma^*D(\theta\bar{\psi}))\in
  S^w_{\sF}([0,t];H^{1,2})\times L^{2}_{\sF}(0,t;H^{1,2}),\quad \forall\,t\in(0,T)
 $$
 and $\bar{u}_t(y)\geq 0$ {a.e.} in $\Omega\times[0,T)\times\bR^d$, we have a.s. for all $t\in[0,T)$,
  $
  \bar{u}_t\geq u_t
  $  a.e. in  $\bR^d$. If we {further} have $p_0>2d+2$ and $\theta\bar{u}\in C^w_{\sF}([0,t];H^{1,p})$ for some $p\in (2d+2,p_0)$, then a.s. for {all} $t\in [0,T)$,   $  \bar{u}_t= u_t $ a.e. in $\bR^d$.
  For this solution, given any $p\in(2,p_0)$ there exists $\alpha\in(1,2)$, s.t.
  $\{(T-t)^{\alpha}(\theta u_t,\theta\psi_t + \sigma^*D (\theta u_t))(y);\,(t,y)\in[0,T]\times \bR^d\}
  $
  belongs to
  $
  \left(S^w_{\sF}([0,T];H^{1,2})\cap C^w_{\sF}([0,T];  H^{1,p})\right)\times L^{2}_{\sF}(0,T;H^{1,2}),
  $
    and  there exist two constants $c_0>0$ and $c_1>0$  s.t. a.s.
  \begin{equation*}
    \frac{c_0}{T-t}\leq u_t \leq \frac{c_1}{T-t}\quad\ {a.e.}\ in\,\bR^d, \, \forall\,t\in[0,T).
  \end{equation*}
Moreover, if the constant $p_0$ introduced in (H.3) satisfies $p_0>2d+2$, then:
\begin{equation*}
    V(t,y,x){\triangleq}u_t(y)x^2,\quad (t,x,y)\in[0,T]\times \bR\times \bR^d,
 \end{equation*}
  coincides with the value function in \eqref{value-func}, and the optimal (feedback) control is given by \begin{align}
    \left(\xi^{*}_t,\ \rho^*_t(z)\right)=
    \left(\frac{u_t(y_t)x_t}{\eta_t(y_t)},\  \frac{u_t(y_t)x_{t-}}{\gamma_t(z,y_t)+u_t(y_t)}   \right).\label{eq-thm-control}
  \end{align}
\end{thm}

{\begin{rmk} Our main result holds for other nonlinear dependencies of $F$ on $u_t$ under more or less standard assumptions. However, as indicated in \cite{Du_Zhang_DegSemilin2012} (see also Theorem \ref{prop-bspde-DTZ}), due to the lack of the regular estimate on the second unknown variable $\psi$, the nonlinear term $F$ needs to be independent of $\psi$, though a linear dependence on $\psi+\sigma^*Du$ is allowed.
\end{rmk}}

\begin{rmk} 
If all the coefficients $b,\sigma,\lambda,\eta,\gamma$ are deterministic, then the optimal control problem is Markovian and the  BSPDE
\eqref{BSPDE-singlr} reduces to the following parabolic PDE:
\begin{equation}\label{BSPDE-singlr-sec6}
  \left\{\begin{array}{l}
  \begin{aligned}
  -\partial_t u_t(y)=\ &\textrm{tr}\left( \frac{1}{2}\sigma_t\sigma^{*}_t(y)
  D^2 u_t(y)
        \right)+b^{*}_t(y)D u_t(y)-\frac{|u_t(y)|^2}{\eta_t(y)}+\lambda_t(y) \\
        &\,-
        \!\int_{\mathcal{Z}}\!\frac{|u_t(y)|^2}{\gamma(t,y,z)+u_t(y)}\mu(dz)
, \quad (t,y)\in [0,T]\times \bR^d;\\
    u_T(y)=\,& +\infty,  \quad y\in\bR^d,
    \end{aligned}
  \end{array}\right.
\end{equation}
where $\sigma_t\sigma^{*}_t$ could be degenerate. Under {Conditions} (H.1)--(H.3), this PDE holds in the distributional (or weak) sense. As such, our results are new even in the Markovian case.
\end{rmk}

\begin{rmk}\label{rmk-gradient-estimate}
{The main novelty of Theorem \ref{thm-conclude} is the gradient estimate for the solution of BSPDE \eqref{BSPDE-singlr}. 
  Gradient estimates for solutions to \emph{degenerate} BSPDEs or even deterministic PDEs generally depend on both its gradient at terminal time and the gradients of the coefficients. For instance, let us consider a trivial version of degenerate PDE \eqref{BSPDE-singlr-sec6}
  $$
  -\partial_t v_t(y)=\lambda_t(y), \quad (t,y)\in [0,T]\times \bR^d;\quad v_T(y)=G(y),\quad y\in\bR^d,
  $$
  for some $G\in H^{1,2}(\bR^d)$. Then we have 
  $$Dv_t(y)=DG(y)+\int_t^TD\lambda_t(y)\,dt,\quad (t,y)\in [0,T]\times \bR^d,$$
  and thus for any $(t,y)\in [0,T]\times \bR^d$,
  $$(T-t)^{\alpha}D(\theta v_t) (y)=(T-t)^{\alpha}v_t (y)D\theta(y) 
  +(T-t)^{\alpha}\theta(y)\left(DG(y)+\int_t^TD\lambda_t(y)\,dt\right).$$
  However, in our case, the terminal value of the BSPDE is singular which does in no obvious way characterize the gradient. For instance, for any given positive Schwartz function $g$ and positive real numbers $r$ and $q$, the function $\varphi$ defined by
   $$\varphi_t(y)\triangleq\frac{r\, g\left((T-t)^{q}y\right)}{(T-t)^{q}},\quad (t,y)\in [0,T]\times \bR^d$$
  shares the singular terminal condition $\varphi_T(y)=+\infty$, $y\in\bR^d$, but its gradient $D\varphi_T(y)=r Dg(0)$ changes according to  the specific choice of the pair $(r,g)$. Hence, in this sense, our gradient estimate is nontrivial, and seems to be new even for the deterministic case. 
  }
\end{rmk}





\section{Existence of a nonnegative solution to BSPDE \eqref{BSPDE-singlr}}\label{sec-existence}


In this section we establish the following existence of solutions result for our BSPDE \eqref{BSPDE-singlr}. 

\begin{thm}\label{thm-existence}
Under Conditions (H.1)-(H.3), for any $p\in [2,p_0)$  BSPDE \eqref{BSPDE-singlr} has a solution $(u,\psi)$ s.t.  for some $\alpha\in(1,2)$,
  $\{(T-t)^{\alpha}(\theta u_t,\theta\psi_t + \sigma^*D (\theta u_t))(y);\,(t,y)\in[0,T]\times \bR^d\} $ belongs to
 $$  \left(S^w_{\sF}([0,T];H^{1,2}) \cap C^w_{\sF}([0,T]; H^{1,p})\right) \times L^{2}_{\sF}(0,T;H^{1,2}),$$
  and a.s.
  \begin{align}
    \frac{c_0}{T-t}\leq u_t(y) \leq \frac{c_1}{T-t}\quad{\textrm{a.e. in } \bR^d,\ \ \forall t\in[0,T),}\label{eq-thm-2-u}
  \end{align}
  with two constants $c_0>0$ and $c_1>0$.
\end{thm}

To prove the above theorem we shall identify a solution as an accumulation point of a convex combination subsequence of the sequence of solutions to BSPDE (\ref{BSPDE-singlr-N}).  For $p_1\in [2,p_0)$ this BSPDE has a unique solution $(v^N,\zeta^N)$, due to Proposition \ref{prop-N}. The sequence $\{v^N\}$ increases in $N$, due to Corollary \ref{cor-prop-N} and hence converges to some limit $v$. To see that $v$ satisfies the growth condition (\ref{eq-thm-2-u}) we replace the coefficients $(\lambda,\gamma,\eta)$ by their lower bound $(0,0,\kappa)$ and upper bound $(\Lambda,+\infty,\Lambda)$, respectively, deduce from Proposition \ref{prop-N} that the resulting BSPDEs have unique solutions given by 
 $$
 \hat{u}^N_t(y){\triangleq}
 \frac{\kappa\mu(\mathcal{Z})\theta(y)}
    {\Big(1+\frac{\kappa\mu(\mathcal{Z})}{N}\Big)e^{\mu(\mathcal{Z})(T-t)}-1}
 \quad  \textrm{and}\quad
 \tilde{u}^N_t(y){\triangleq}\frac{2\Lambda\theta(y)}
 {1-\frac{N-\Lambda}{N+\Lambda}\cdot e^{-2(T-t)}}-\Lambda \theta(y){,}
 $$
 and then apply the comparison principle to conclude that a.s.:
 \[
 	\hat{u}^N_t(y) \leq v^N_t(y) \leq  \tilde{u}^N_t(y)\quad{\textrm{a.e. in }\bR^d,\ \ \forall t\in[0,T).}
 \]
 Since
 $$
 \tilde{u}^N_t
 \leq
 \frac{2\Lambda\theta(y)}
 {1-\frac{N-\Lambda}{N+\Lambda}\cdot e^{-2(T-t)}}
 =
 \frac{2\Lambda\theta(y)e^{2(T-t)}}
 {e^{2(T-t)}-\frac{N-\Lambda}{N+\Lambda}}
 \leq
 \frac{2\Lambda\theta(y)e^{2(T-t)}}
 {1+2(T-t)-\frac{N-\Lambda}{N+\Lambda}}
 \leq
 \frac{\theta(y)e^{2T}}
 {\frac{1}{N+\Lambda}+\frac{T-t}{\Lambda}}
 $$
we see that
  \begin{align}
    \frac{\kappa \mu(\mathcal{Z})\theta(y)}
    {\Big(1+\frac{\kappa\mu(\mathcal{Z})}{N}\Big)e^{\mu(\mathcal{Z})(T-t)}-1}
\leq v^N_t(y)
    \leq
     \frac{\theta(y)e^{2T}}{\frac{1}{N+\Lambda}+\frac{T-t}{\Lambda}}\quad \textrm{a.e. in }\bR^d\label{est-prf-thm-exis-N-T0}
  \end{align}
and hence that $v$ satisfies the desired growth condition. The slightly sharper upper bound ${\cal O}\left(\frac{1}{N^{-1} + T-t}\right)$ for $v^N_t$ will be important for the proof of Lemma \ref{lem-grad-esti} below and hence for the gradient estimate. 

Our next goal is to prove a uniform bound for  the sequence $\{D v^N\}$ in $H^{0,p}$. As a byproduct we obtain a bound for the sequence  $\{\zeta^N + \sigma^* D v^N\}$ in $\cL^2(0,T;H^{1,2})$. 
The bound given in Theorem \ref{prop-bspde-DTZ} (ii) depends on the Lipschitz constant of the driver of the BSPDE. In our case, this means that it depends on the function $v^N$, due to the quadratic dependence of the driver on $v^N$. The following corollary provides a better estimate. The estimates in Theorem \ref{prop-bspde-DTZ} are obtained by applying It\^o formulas directly (see \cite{DuTangZhang-2013,Du_Zhang_DegSemilin2012}); hence we can derive the estimates {as well} from the monotonicity of the drift for BSPDE \eqref{BSPDE-dqiutang} instead of the Lipschitz condition. The detailed proof is omitted; it is standard but cumbersome.
\begin{cor}\label{cor-appdx-estim}
  Assume the same hypothesis of Theorem \ref{prop-bspde-DTZ} with $$f(\cdot,\cdot,0)\in \cL^p_{\sF}(0,T;H^{1,p} )\cap \cL^2_{\sF}(0,T; H^{1,2}) \textrm{ and }
  G\in L^p(\Omega,\sF_T;H^{1,p})\cap L^2(\Omega,\sF_T; H^{1,2})
  $$
    for some $p\in [2,\infty)$. Let $(u,\phi)$ be the solution of BSPDE \eqref{BSPDE-dqiutang} in Theorem \ref{prop-bspde-DTZ}. If there exist constant $L_1$ and function $g\in\cL^p_{\sF}(0,T;H^{0,p})\cap \cL^2_{\sF}(0,T;L^2)$ s.t. {a.e.} in $ \Omega\times[0,T]\times\bR^d$,
    \begin{align}
    &u_s(y)f(s,y,u_s(y))+\sum_{i=1}^d\partial_{y^i}u_s(y) \left(\partial_{y^i}+\partial_{y^i}u_s(y)\partial_{u}\right)f(s,y,u_s(y))
    \nonumber\\
    &\leq
    |g_s(y)|^2+L_1\Big( |u_s(y)|^2+ \sum_{i=1}^d|\partial_{y^i}u_s(y)|^2 \Big),
    \label{relation-cor-est-appdx}
    \end{align}
then we have
$$
    {E}\sup_{t\in[0,T]}\|u_t\|_{H^{1,p}}^{p}
    \leq C'_p E\bigg[  \|G\|_{H^{1,p}}^p +\int_0^T\!\! \|g_t\|_{H^{0,p}}^p dt  \bigg],
$$
and
  \begin{eqnarray*}
    {E}\sup_{t\in[0,T]}\|u_t\|_{H^{1,2}}^{2} + {E}\int_{0}^{T}\!\!\|\phi_t +
    \sigma_t^* Du_t\|_{H^{1,2}}^{2}\,dt
    \leq C_2'\, {E} \bigg[\|G\|_{H^{1,2}}^{2}
    + \int_{0}^{T}\!\!\|g_t\|^{2}dt\bigg]
    \end{eqnarray*}
with the constants $C'_2=C'_2(d,m,\Lambda,L,T,L_1)$ and $C'_p=C'_p(d,m,\Lambda,L,T,L_1,p)$ independent of the Lipchitz constant $L_0$.
\end{cor}

We  proceed with the gradient estimate. Since we are mainly interested in the behavior of the gradient near the terminal time, we put
\[
 	\kappa_1{\triangleq}\essinf_{(\omega,t,y)\in\Omega \times [T_0,T]\times\bR^d} \eta_t(y)
\]
and notice that (\ref{est-prf-thm-exis-N-T0}) holds with $\kappa$ replaced by $\kappa_1$ on $[T_0,T]$. {The following lemma is key to the gradient estimate. }

\begin{lem} \label{lem-grad-esti}
Recall the constant $p_0$ introduced in (H.3), let $\alpha_0\triangleq1-\frac{1}{2p_0}$, and choose $\alpha_1,\alpha_2\in (1,\infty)$ and $p_1 \in [2,p_0)$ s.t.
  $$
  2\alpha_0=\alpha_1\alpha_2  \quad \textrm{ and }\quad (2-\alpha_2)p_1<1.
  $$
 Let $T_1\in[T_0,T)$ and $N_0>2\Lambda+{\kappa}\mu(\mathcal{Z})$ s.t.
  $$
  \bigg(1+\frac{\kappa_1\mu(\mathcal{Z})}{N_0}\bigg)e^{\mu(\mathcal{Z})(T-T_1)}
  <\alpha_1,
  $$
and for each $N>N_0$, set
  $${\delta^N}{\triangleq}\bigg(1+\frac{\kappa_1\mu(\mathcal{Z})}{N}\bigg)
  e^{\mu(\mathcal{Z})(T-T_1)}.$$
 Then the  sequence
 $$
(Q^N_t,\xi^N_t){\triangleq}\bigg(\frac{\kappa_1}{N}+{\delta^N}(T-t)\bigg)^{\alpha_2}(v^N_t,\zeta^N_t), \quad \textrm{$t\in[0,T]$},
$$
satisfies
 \begin{align*}
  \sup_{N>N_0}
  \bigg\{
  E\Big[\sup_{t\in[T_1,T]}
  \left(\|Q^N_t\|^2_{H^{1,2}}+\|Q^N_t\|^{p_1}_{H^{1,{p_1}}}\right)\Big]
  +
  \|\sigma^*DQ^N+\xi^N\|^2_{\cL^2(T_1,T;H^{1,2})}
  \bigg\}<\infty{.}
\end{align*}
\end{lem}
\begin{proof}
A direct computation shows that the sequence $\{(Q^N,\xi^N)\}$ is a solution to the BSPDE:
  \begin{equation}\label{BSPDE-singlr-N-tmTrfm}
  \left\{\begin{array}{l}
  \begin{aligned}
  -dQ^N_t(y)=\, &\bigg[\textrm{tr}\left( \frac{1}{2}\sigma_t\sigma^*_tD^2 Q^N_t(y)
        + D\xi^N_t\sigma^{*}_t(y)\right)
        +\tilde{b}^{*}_tD Q^N_t(y)
+\beta_t^{*}\xi^N_t(y)+c_tQ^N_t(y)\\
       +\bigg(&\frac{\kappa_1}{N}+{\delta^N}(T-t)\bigg)^{\alpha_2}
        \bigg(
        \theta\lambda_t(y)
        -\int_{\mathcal{Z}}\frac{\theta^{-1} |v^N_t(y)|^2}{\gamma_t(y,z)+\theta^{-1} |v^N_t(y)|}\mu(dz)
        - \frac{\theta^{-1}\left|v^N_t(y)\right|^2}{\eta_t(y)}
         \bigg)\\
         +&\alpha_2{\delta^N}
         \bigg(\frac{\kappa_1}{N}+{\delta^N}(T-t)\bigg)^{\alpha_2-1}v^N_t(y)
        \bigg]\, dt-\xi^N_t(y)\, dW_{t}, \quad (t,y)\in [0,T]\times \bR^d;\\
    Q^N_T(y)=\,& \kappa_2^{\alpha_1}N^{1-\alpha_2}\theta(y)  \quad \textrm{for $y\in\bR^d$}.
    \end{aligned}
  \end{array}\right.
\end{equation}

The assertion follows if we can show that this BSPDE satisfies the the conditions of Corollary  \ref{cor-appdx-estim} on $[T_1,T]$ with some constant $L_1 < \infty$ independent of $N$ and a function $g_N\in\cL^p(T_1,T;H^{0,p})$ which satisfies
\begin{align*}
 \sup_{N>N_0}\|g_N\|_{\cL^{p}_{\sF}(T_1,T;H^{0,p})} <\infty \quad \text{for } p\in\{2,p_1\}.
 \end{align*}
To obtain the desired result it suffices to estimate $\partial_{y^i}Q_t^N(y)\partial_{y^i} (f^2_N-f^1_N)(s,y)$ where
$$
f_N^1(t,y)\,{\triangleq}\,\bigg(\frac{\kappa_1}{N}+{\delta^N}(T-t)\bigg)^{\alpha_2}
\frac{\theta^{-1}\left|v^N_t(y)\right|^2}{\eta_t(y)}
\quad\textrm{and}\quad
f_N^2(t,y)\,{\triangleq}\,\alpha_2{\delta^N}
\bigg(\frac{\kappa_1}{N}+{\delta^N}(T-t)\bigg)^{\alpha_2-1}v^N_t(y).
$$
To this end, notice that $\delta^N< \alpha_1$ and that $e^x\leq 1+xe^x$ for any $x\geq 0$. Hence, for each $N>N_0$, each $t\in[T_1,T)$ and  almost every $y\in\bR^d$ one has:
\begin{align}
    0\geq\,{\alpha_2{\delta^N}}\bigg(\frac{\kappa_1}{N}+{\delta^N}(T-t)\bigg)^{-1}
    -\frac{2\theta^{-1}v^N(y)}{\eta_t(y)}, \label{eq-key-monotn}
\end{align}
because
  \begin{align}
    &{\alpha_2{\delta^N}}\bigg(\frac{\kappa_1}{N}+{\delta^N}(T-t)\bigg)^{-1}
    -\frac{2\theta^{-1}v^N(y)}{\eta_t(y)}\nonumber\\
    \leq &\,
    \alpha_1\alpha_2\bigg(\frac{\kappa_1}{N}+{\delta^N}(T-t)\bigg)^{-1}
    -\frac{2\kappa_1\mu(\mathcal{Z})/\eta_t(y)}
    {\Big(1+\frac{\kappa_1\mu(\mathcal{Z})}{N}\Big)e^{\mu(\mathcal{Z})(T-t)}-1}
    \nonumber\\
    \leq &\,
    \alpha_1\alpha_2\bigg(\frac{\kappa_1}{N}+{\delta^N}(T-t)\bigg)^{-1}
    -{2\alpha_0}{\mu(\mathcal{Z})}
    \bigg( \Big(1+\frac{\kappa_1\mu(\mathcal{Z})}{N}\Big)e^{\mu(\mathcal{Z})(T-t)}-1  \bigg)^{-1}\nonumber\\
    \leq &\,0. \label{eq-monotone-verfy}
 \end{align}
Thus, for any $t\in[T_1,T]$, we have
\begin{align*}
  &\partial_{y^i}Q_t^N(y)\partial_{y^i}
  (f^2_N-f^1_N)(s,y)
  \\
  =&\,\partial_{y^i}Q_t^N(y)
  \alpha_2{\delta^N}
\bigg(\frac{\kappa_1}{N}+{\delta^N}(T-t)\bigg)^{-1}\partial_{y^i}Q^N_t(y)
  -\partial_{y^i}Q_t^N(y)
\frac{2\theta^{-1}v^N_t(y)\partial_{y^i}Q_t^N}{\eta_t(y)}
+\partial_{y^i}Q_t^N(y)f_N^3(t,y)
\\
=&\,
\left|\partial_{y^i}Q_t^N(y)\right|^2\bigg(
  \alpha_2{\delta^N}
\bigg(\frac{\kappa_1}{N}+{\delta^N}(T-t)\bigg)^{-1}
  -
\frac{2\theta^{-1}v^N_t(y)}{\eta_t(y)}\bigg)
+\partial_{y^i}Q_t^N(y)f_N^3(t,y)
\\
\leq&\,\left|\partial_{y^i}Q_t^N(y)\right|^2+\left|f_N^3(t,y)\right|^2 \quad
\text{a.e. in $\bR^d$ a.s.},
\end{align*}
where $$f_N^3(t,y)\triangleq
\bigg(\frac{\kappa_1}{N}+{\delta^N}(T-t)\bigg)^{\alpha_2}
 \frac{\theta^{-1}\left|v^N_t(y)\right|^2}{\eta_t(y)}
 \bigg(\frac{\partial_{y^i}\eta_t(y)}{\eta_t(y)}
 -\partial_{y^i}\theta^{-1}(y)\theta(y) \bigg).$$
In view of the upper bound in \eqref{est-prf-thm-exis-N-T0} there exists a constant $C < \infty$ s.t.
\[
	\theta^{-1}|v^N_t(y)|^2 \leq C \left( \frac{1}{\frac{1}{N}+ T-t} \right)^2\theta(y).
\]
Since $(2-\alpha_2)p_1<1$, one therefore has  for $p\in\{2,p_1\}$ that
\begin{align*}
\sup_{N>N_0}\|f_N^3\|_{\cL^{p}_{\sF}(T_1,T;H^{0,p})}
<\infty.
\end{align*}
This proves the assertion.
\end{proof}


\begin{cor}
The previous lemma{,} along with
an application of Theorem \ref{prop-bspde-DTZ} to the time interval $[0,T_1]$, leads to the desired gradient estimate:
\begin{align}
  \sup_{N>N_0}
  \bigg\{
  E\Big[\sup_{t\in[0,T]}
  \left(\|Q^N_t\|^2_{H^{1,2}}+\|Q^N_t\|^{p_1}_{H^{1,{p_1}}}\right)\Big]
  +
  \|\sigma^*DQ^N+\xi^N\|^2_{\cL^2(0,T;H^{1,2})}
  \bigg\}<\infty.\label{bounded-est-prf-exst}
\end{align}
\end{cor}

We are now ready to prove existence of a solution to our BSPDE \eqref{BSPDE-singlr}. Estimate (\ref{bounded-est-prf-exst}) allows us to extract a subsequence  $(Q^{N_k},\xi^{N_k})$ s.t.
$Q^{N_k}$ converges to  $Q$ weakly in  $\cL^{p}(0,T;H^{1,p})$ as well as weak-star in $\cL^{\infty}(0,T;H^{1,p})$ for any $p\in\{2,p_1\}$, and  $(\xi^{N_k},\xi^{N_k}+\sigma^*DQ^{N_k})$ converges weakly to  $(\xi,\xi+\sigma^*DQ)$ in
$\cL^{2}(0,T;L^2)\times \cL^{2}(0,T;H^{1,2})$.
Since $\{v^N\}$ increases to $v$ {a.e.} in $\bR^d$ for all $t\in[0,T]$, passing to the limit we get
$$Q_t(y)=e^{\alpha_2\mu(\mathcal{Z})(T-T_1)}(T-t)^{\alpha_2}v_t(y).$$
Mazur's Lemma allows us to choose a sequence of convex combinations of $(Q^{N_k},\xi^{N_k},\xi^{N_k}+\sigma^*DQ^{N_k})$ {which} converges strongly in corresponding spaces. Therefore, it is easy to check that $(Q,\xi)$ solves:
  \begin{equation}\label{BSPDE-singlr-N-tmTrfm-limit}
  \left\{\begin{array}{l}
  \begin{aligned}
  -dQ_t(y)=\, &\bigg[\text{tr}\left( \frac{1}{2}\sigma_t\sigma_t^*
  D^2 Q_t(y)
        + D \xi_t\sigma^{*}_t(y)\right)
        +\tilde{b}^{*}_tD Q_t(y)
+\beta_t^{*}\xi_t(y)+c_tQ_t(y)\\
        +e&^{\alpha_2\mu(\mathcal{Z})(T-T_1)}(T-t)^{\alpha_2}
        \bigg(
        \theta\lambda_t(y)
        -\int_{\mathcal{Z}}\frac{\theta^{-1} |v_t(y)|^2}{\gamma_t(y,z)
        +\theta^{-1} |v_t(y)|}\mu(dz)
        - \frac{\theta^{-1}\left|v_t(y)\right|^2}{\eta_t(y)}
         \bigg)\\
         &\,  +\alpha_2
         e^{\alpha_2\mu(\mathcal{Z})(T-T_1)}(T-t)^{\alpha_2-1}v_t(y)
        \bigg]\, dt-\xi_t(y)\, dW_{t}, \quad (t,y)\in [T_1,T]\times \bR^d;\\
    Q_T(y)=\,&0,  \quad y\in\bR^d.
    \end{aligned}
  \end{array}\right.
\end{equation}
By Theorem \ref{prop-bspde-DTZ} and Proposition \ref{prop-comprn}, $(Q,\xi)$ admits a version, still denoted by $(Q,\xi)$, s.t.
$$(Q,\xi+\sigma^*DQ)\in \left(S^w_{\sF}([0,T];H^{1,2})\cap C^w_{\sF}([0,T]; H^{1,p_1})\right)\times L^{2}_{\sF}(0,T;H^{1,2}).$$
Recovering $(v,\zeta)$ from $(Q,\xi)$ and setting $(u,\psi)\,{\triangleq}\,\theta^{-1} (v,\zeta)$ we see that
%
$(u,\psi)$ solves BSPDE \eqref{BSPDE-singlr} and that
  $\{(T-t)^{\alpha_2}(\theta u_t,\theta\psi_t + \sigma^*D (\theta u_t))(y);\,(t,y)\in[0,T]\times \bR^d\}$
  belongs to
  $\left(S^w_{\sF}([0,T];H^{1,2})\cap C^w_{\sF}([0,T]; H^{1,p_1})\right)\times L^{2}_{\sF}(0,T;H^{1,2}).
  $
  Moreover, relation \eqref{eq-thm-2-u} holds with
$
c_0={\kappa}e^{-\mu(\mathcal{Z})T}$  and $ c_1=\Lambda e^{2T}$.
Since $p_1\in[2,p_0)$ is arbitrary, this completes the proof of Theorem \ref{thm-existence}.

\section{{Verification Theorem and uniqueness of solution to BSPDE \eqref{BSPDE-singlr}}}


{In this section, we prove a verification theorem which not only solves our control problem with constraint but also allows us to derive uniqueness of the solution to our BSPDE \eqref{BSPDE-singlr}. As no generalized It\^o-Kunita formula is available for the random filed $u$ satisfying BSPDE \eqref{BSPDE-singlr} in the distributional sense due to the degeneracy, the proof is instead based on the link between FBSDEs and BSPDEs given in Theorem \ref{prop-bspde-DTZ}. This link will allow us to compute the dynamics of the process $u_t(y_t) | x_t|^2$.} 

First, we recall a result from \cite{GraeweHorstQui13}. It states that the optimal control lies in the set of controls $\mathscr{A}$ for which the corresponding state process is monotone.    
\begin{lem}\label{lem-admissible-control}
  For each admissible control pair $(\xi,\rho)\in \cL^2_{\bar{\sF}}(0,T)\times \cL^{2}_{\bar{\sF}}(0,T;L^2(\mathcal{Z}))$, there exists a corresponding admissible control pair $(\hat{\xi},\hat{\rho})\in \cL^2_{\bar{\sF}}(0,T)\times \cL^{2}_{\bar{\sF}}(0,T;L^2(\mathcal{Z}))$ whose cost is {no} more than that of $(\xi,\rho)$ and for which the corresponding state process
  $x^{0,x;\hat{\xi},\hat{\rho}}$ is a.s. monotone. Moreover, there exists a constant $C>0$ which is independent of the initial data {$(0,x)$, terminal time $T$} and the control pair $(\hat{\rho}, \hat{\xi})$, s.t.
  \begin{align}
E\left[\left.\sup_{s\in[t,T]}|x_s^{0,x;\hat{\xi},\hat{\rho}}|^2\right|\bar{\sF}_t\right]
    =
   |x_t^{0,x;\hat{\xi},\hat{\rho}}|^2
    \leq\, C (T-t)E\left[\left.\int_t^T|\hat{\xi}_s|^2\,ds
    \right|\bar{\sF}_t\right]\quad \textrm{for each}\ t\in[0,T]. \label{est-lem-adm-control}
  \end{align}
\end{lem}

{The key to the verification theorem is existence of a solution $(u,\psi)$ to BSPDE \eqref{BSPDE-singlr} such that $u$ satisfies a growth condition near the terminal time and that its gradient is sufficiently regular (both guaranteed by Theorem \ref{thm-existence}) so that Theorem \eqref{prop-bspde-DTZ}(iii) can be applied and $u$  can be represented as an FBSDE.}

\begin{thm}\label{thm-verification}
Assume (H.1)--(H.2). If $(u,\psi)$ is a solution to the BSPDE \eqref{BSPDE-singlr} s.t. $\theta u\in C^w(0,t;H^{1,p}) \cap S^w(0,t; H^{1,2}) $, $\forall\,t\in(0,T)$, for some $p>2d+2$, and a.s.
  \begin{align}
    \frac{c_0}{T-t}\leq u_t(y) \leq \frac{c_1}{T-t},\quad \forall\,(t,y)\in[0,T)\times\bR^d \label{eq-thm-u}
  \end{align}
  with two constants $c_0>0$ and $c_1>0$, then
  \begin{align}
    V(t,y,x){\triangleq}u_t(y)x^2,\quad \forall
    \,(t,x,y)\in[0,T]\times \bR\times \bR^d,
  \end{align}
  coincides with the value function \eqref{value-func}. Moreover, the optimal feedback control is given by \eqref{eq-thm-control}.
\end{thm}



\begin{proof}
	
We first note that $u_t(y)$ is a.s. continuous with respect to $(t,y)\in[0,T)\times\bR^d$, due to Sobolev's embedding theorem.
Second, BSPDE \eqref{BSPDE-singlr} is equivalent to BSPDE  \eqref{BSPDE-singlr-1}. Thus, if we take $\tau \in (0,T)$ as the terminal time and $\theta u_\tau(y)$ as the terminal condition, then this BSPDE satisfies the assumptions of Theorem \ref{prop-bspde-DTZ} on $[0,\tau]$, due to \eqref{eq-thm-u} ({and the presence of the weight function $\theta$)}. As a result, there exists a unique random field $\psi$ s.t. $\theta\psi+\sigma^* D(\theta u)\in L^{2}_{\sF}(0,\tau;H^{1,2})$ for any $\tau \in(0,T)$ and s.t. $(\theta u,\,\theta \psi)$ is a solution to:
  \begin{equation*}
  \left\{\begin{array}{l}
  \begin{aligned}
  -dv_t(y)=\, &\bigg[\textrm{tr}\left(\frac{1}{2}\sigma_t\sigma_t^{*}
  D^2 v_t(y)
        + D \zeta_t\sigma^{*}_t(y)\right)+b^{*}_tD v_t(y)
+\theta\sum_{ijr}\sigma_t^{jr}\partial_{y^j}\theta^{-1}\left(\zeta_t^r
+\partial_{y^i}v\sigma_t^{ir}\right)(y)
\\
        &\,
        +\theta \bigg(
\frac{1}{2}\sum_{ijr}\partial_{y^iy^j}\theta^{-1}\sigma_t^{ir}\sigma_t^{jr}
+b^*_tD\theta^{-1}
\bigg)v_t(y)
        +\theta(y)F(t,y,\theta^{-1}(y)v_t(y))   \bigg]\, dt\\
       &\, -\zeta_t(y)\, dW_{t}, \quad (t,y)\in [0,\tau)\times \bR^d;\\
    v_{\tau}(y)=\,& \theta u_{\tau}(y),  \quad y\in\bR^d.
    \end{aligned}
  \end{array}\right.
\end{equation*}

{Noticing the assumption $p>2d+2$, by Theorem \ref{prop-bspde-DTZ} (iii) we} also have the following BSDE representation of $\theta u$:
\begin{align*}
  -d(\theta u_t)(y_t^{0,y})
  = &\,\bigg[
  \theta \bigg(
\frac{1}{2}\sum_{ijr}\partial_{y^iy^j}\theta^{-1}\sigma_t^{ir}\sigma_t^{jr}
+b^*_tD\theta^{-1}
\bigg)(\theta u_t)(y_t^{0,y})
+\theta\sum_{jr}\partial_{y^j}\theta^{-1}\sigma_t^{jr}(y_t^{0,y})\big(Z_t^{0,y}\big)^r
\\
&\,\,+\theta(y_t^{0,y})F(t,y_t^{0,y},u_t(y_t^{0,y}))
  \bigg]\,dt
  -Z_t^{0,y}\,dW_t,\quad t\in [0,T)
\end{align*}
 for some adapted process $Z^{0,y}$ lying in suitable space. Applying the standard It\^o formula, we obtain
\begin{align*}
  d\theta^{-1}(y_t^{0,y})
  =\bigg[\frac{1}{2}\textrm{tr}\,\left(\sigma_t\sigma_t^*D^2\theta^{-1}(y_t^{0,y})\right)
  +b_t^*D\theta^{-1}(y_t^{0,y})
  \bigg]\,dt
  +(D\theta^{-1})^*\sigma_t(y_t^{0,y})\,dW_t,
\end{align*}
and further,
\begin{align*}
  -d u_t(y_t^{0,y})
  = &\,
  F(t,y_t^{0,y},u_t(y_t^{0,y}))
  \,dt
  -\left[\theta^{-1}(y_t^{0,y})Z_t^{0,y}+\theta u_t(D\theta^{-1})^*\sigma_t(y_t^{0,y})\right]\,dW_t,\quad t\in [0,T).
\end{align*}
{Then the stochastic differential equation for $u_t(y_t^{0,y})|x_t^{0,x;\xi,\rho}|^2$ follows immediately from an application of the standard It\^o formula again.} Using Lemma \ref{lem-admissible-control} one can now apply the exact same arguments as in the proof of \cite[Theorem 3.1]{GraeweHorstQui13} to deduce that:
  \begin{align*} \label{eq-thm-prf-1}
    u_t(y_t^{0,y})|x_t^{0,x;\xi,\rho}|^2
    \leq J(t,x_t^{0,x;\xi,\rho},y_t^{0,y};\xi,\rho) \quad \text{for any pair} \quad (\xi,\rho)\in\mathscr{A}
  \end{align*}
and that the control $(\xi^*,\rho^*)$ is admissible and satisfies the above inequality with equality.
\end{proof}


We close our analysis with the following theorem. It states that the solution constructed in Section \ref{sec-existence} is the minimal solution to our BSPDE \eqref{BSPDE-singlr}. The proof mainly relies on the comparison principle in Proposition \ref{prop-comprn} for degenerate BSPDEs allowing for singular terminal values.

\begin{thm}\label{thm-minimal}
Under Conditions (H.1)--(H.3), for the solution $(u,\psi)$ to BSPDE \eqref{BSPDE-singlr} constructed in the proof of Theorem \ref{thm-existence}, if $(\tilde{u},\tilde{\psi})$ is another solution of \eqref{BSPDE-singlr} satisfying 
 $$
(\theta\tilde{u},\theta\tilde{\psi}+\sigma^*D(\theta\tilde{\psi}))\in
  S^w_{\sF}([0,t];H^{1,2})\times L^{2}_{\sF}(0,t;H^{1,2}),\quad \forall\,t\in(0,T)
 $$
and if $\tilde{u}_t(y)\geq 0$ {a.e.} in $\Omega\times [0,T)\times\bR^d$, then a.s. for every $t\in [0,T)$,
  $  \tilde{u}_t\geq u_t $ a.e. in $\bR^d$. Moreover, if we {further} have $p_0>2d+2$ and $\theta\tilde{u}\in \cap_{t\in(0,T)} C^w_{\sF}([0,t];H^{1,p})$ for some $p\in (2d+2,p_0)$, then a.s. for {all} $t\in [0,T)$,
  $  \tilde{u}_t= u_t $ a.e. in $\bR^d$.
\end{thm}

\begin{proof}
Let $(v^N,\zeta^N)$ be the unique solution to BSPDE \eqref{BSPDE-singlr-N} and 
$(\tilde{v},\tilde{\zeta})\,\triangleq\,\theta\,(\tilde{u},\tilde{\psi})$. By Proposition \ref{prop-comprn}, 
  \begin{equation}\label{eqf-thm-minimal}
    \tilde{v}_t\geq v^N_t\quad \text{a.e. in }\bR^d,\ \ \forall t\in[0,T],
  \end{equation}
  which yields the minimality as $v^N$ {increases} to $v$ as $N\to\infty$.
	In view of Theorem~\ref{thm-verification}, to establish the uniqueness statement it is sufficient to verify that $\tilde {u}$ satisfies the growth condition \eqref{eq-thm-u}. The above minimality arguments have given the lower bound. To establish the upper bound in \eqref{eq-thm-u}, we consider the deterministic function:
\[
	\hat u_t\,{\triangleq}\,\Lambda\coth(T-t)=\frac{2\Lambda}{1-e^{-2(T-t)}}-\Lambda\leq \frac{\Lambda e^{2T}}{T-t}.
\]
Then, $(\hat u, 0)$ is a solution to \eqref{BSPDE-singlr} with the triple $(\lambda,\gamma,\eta)$ being replaced by $(\Lambda,+\infty,\Lambda)$. Moreover, $(\hat u, 0)$ remains a solution  when shifted in time, i.e., for $\delta\in[0,T)$ the pair $(\hat u_{\,\cdot\,+\delta} ,0)$ is the solution to \eqref{BSPDE-singlr} associated with $(\Lambda,+\infty,\Lambda)$, but with a singularity at $t=T-\delta$. Proposition \ref{prop-comprn} yields that, a.s. for all $t\in[0,T-\delta]$
\[
	\tilde u_t\leq\frac{\Lambda e^{2T}}{T-\delta-t} \quad \ \ \text{a.e. in $\bR^d$.}
\]
Letting $\delta\rightarrow 0$ we obtain the desired upper bound as well as the uniqueness.
\end{proof}


\begin{appendix}

\section{Selected results on semi-linear degenerate BSPDEs}

{This appendix recalls the selected results on degenerate semi-linear BSPDEs and their connections to FBSDEs, establishes a comparison principle for degenerate BSPDEs allowing for singular terminal values and discusses a truncated version of our singular BSPDE.} 

\subsection{{ On a class of degenerate BSPDEs and the link to FBSDEs}}

The following link between FBSDEs and BSPDEs, due to \cite{Du_Zhang_DegSemilin2012} is key to our analysis. 

\begin{thm}\label{prop-bspde-DTZ}
Assume that the coefficients $b$ and $\sigma$ satisfy (H.1) and (H.2) and that $\varrho: \Omega\times [0,T]\times\bR^d\rightarrow \bR^m$ satisfies the same conditions as $b$. Let
$f:\Omega\times[0,T]\times\bR\rightarrow \bR$  satisfy:
\begin{itemize}
	\item the partial derivatives $\partial_yf$ and $\partial_vf$ exist  for any quadruple $(\omega,t,y,v)$
	\item $f(\cdot,\cdot,0) \in \cL^{2}_{\sF}(0,T;H^{1,2})$
	\item there exists a constant $L_0>0$ s.t. for each $(\omega,t,y)$,
\begin{eqnarray*}
|f(t,y,v_1)-f(t,y,v_2)|+|\partial_yf(t,y,v_1)-\partial_yf(t,y,v_2)|\leq L_0 |v_1 - v_2|,\quad
\forall~v_1,v_2\in \bR.
\end{eqnarray*}
\end{itemize}
Then the following holds:
\begin{itemize}\item[i)]
For any $G\in L^2(\Omega,\sF_T;H^{1,2})$, the BSPDE
\begin{equation}\label{BSPDE-dqiutang}
  \left\{\begin{array}{l}
  \begin{aligned}
  -du_t(y)=\, &\bigg[\textup{tr}\left( \frac{1}{2}\sigma_t\sigma^{*}_tD^2u_t
+D \psi_t\sigma^{*}_t\right)(y)+b^{*}_tD u_t(y)+\varrho^*_t\left(\psi_t+\sigma_t^* Du_t \right)(y)
        +f(t,y,u_t)\bigg]\, dt\\
        &-\psi_t(y)\, dW_{t}, \quad (t,y)\in [0,T]\times \bR^d;\\
    u_T(y)=\,& G(y),  \quad y\in\bR^d
    \end{aligned}
  \end{array}\right.
\end{equation}
admits a unique solution  $(u,\psi)$ s.t.
\[
  	u\in S^w_{\sF}([0,T];H^{1,2}) \quad \text{and} \quad \psi+\sigma^* Du\in L^{2}_{\sF}(0,T;H^{1,2}).
\]
Moreover, there exists a constant $C_2=C_2(d,m,\Lambda,L,T,L_0)$
  s.t.
  \begin{eqnarray}\label{est:sl}
    {E}\sup_{t\in[0,T]}\|u_t\|_{H^{1,2}}^{2} + {E}\int_{0}^{T}\!\!\|\psi_t +
    \sigma_t^*Du_t\|_{H^{1,2}}^{2}\,dt
    \leq C_2\, {E} \bigg[\|G\|_{H^{1,2}}^{2}
    + \int_{0}^{T}\!\!\|f(t,\cdot,0)\|_{H^{1,2}}^{2}\,dt\bigg].\ \ \ \
    \end{eqnarray}

\item[ii)] If we further assume that $f(\cdot,\cdot,0) \in \cL^{p}_{\sF}(0,T;H^{1,p})$ and $G\in L^p(\Omega,\sF_T;H^{1,p})$ for some $p\in[2,\infty)$, then $u\in C^w_{\sF}([0,T];H^{1,p})$ and there exists a constant $C_p=C_p(d,m,\Lambda,L,T,L_0,p)$ s.t.
\begin{eqnarray}\label{HQZ3}
    {E}\sup_{t\in[0,T]}\|u_t\|_{H^{1,p}}^{p}
    \leq C_pE\bigg[  \|G\|_{H^{1,p}}^p +\int_0^T\!\! \|f(t,\cdot,0)\|_{H^{1,p}}^p dt  \bigg].
\end{eqnarray}

\item[iii)] If $p>2d+2$, then $u(t,y)$ is a.s. continuous with respect to $(t,y)$ and it holds a.s. that
      \begin{eqnarray}
  u(t,y^{s,y}_t)=Y^{s,y}_t,\quad \forall\, (t,y)\in[s,T]\times\bR^d,
  \end{eqnarray}
  where $(y_{\cdot}^{s,y},Y_{\cdot}^{s,y},Z_{\cdot}^{s,y})$ is the solution of FBSDE:
  \begin{equation*}
  \left\{\begin{array}{l}
  \begin{aligned}
  dy_t^{s,y}=&b_t(y_t^{s,y})\,dt+\sigma_t(y_t^{s,y})\,dW_t,\quad y_s^{s,x}=y;\quad 0\leq s \leq t\leq T;\\
  -dY_t^{s,y}=&\left[\varrho^*_t(y_t^{s,y})Z_t^{s,y}+f(t,y_t^{s,y},Y_t^{s,y})\right]dt
  -Z_t^{s,y}\,dW_t,\quad Y_T^{s,y}=G(y_T^{s,y}).
\end{aligned}
  \end{array}\right.
\end{equation*}
\end{itemize}
\end{thm}

\vspace{2mm}

\begin{rmk}
It is worth noting that assertion (i) of the above theorem extends \cite[Theorem 3.1]{Du_Zhang_DegSemilin2012} by replacing $C^w_{\sF}([0,T];H^{1,2})$ therein by $S^w_{\sF}([0,T];{H^{1,2}}) = C^w_{\sF}([0,T];H^{1,2})\cap L^2(\Omega,\sF;C([0,T];L^2))$. This follows by applying  \cite[Theorem 3.2]{QiuTangBDSDES2010} with the Gelfand triple being realized as $(H^{-1,2},L^2,H^{1,2})$ therein\footnote{ $H^{-1,2}$ is the dual space of $H^{1,2}$.}. In particular, we obtain strong continuity of $u$ in $L^2$. This allows us to apply the existing It\^o formula for BSPDEs (see for instance \cite[Corollary 3.11]{QiuWei-RBSPDE-2013}).
\end{rmk}

\subsection{A comparison principle for degenerate BSPDEs allowing for singular terminal values}
The following comparison principle follows from the It\^{o} formula for the square norm of the positive part of solution to a BSPDE (see \cite[Lemma A.3]{GraeweHorstQui13}, \cite[Theorem 3.2]{QiuTangBDSDES2010}, \cite[Corollary 3.11]{QiuWei-RBSPDE-2013}). {Since it allows the associated BSPDEs to have singular terminal values, we provide a short proof.}

{
\begin{prop}\label{prop-comprn}
Assume that the coefficients $b$, $\sigma$ and $\varrho$
satisfy the conditions of Theorem \ref{prop-bspde-DTZ}. Let $(u,\psi)$ and $(u',\psi')$ be solutions to BSPDE \eqref{BSPDE-dqiutang} associated with $(G,f)$ and $(G',f')$, respectively, such that 
$(u,\psi+\sigma^*Du, f(u))\in\cL_{\sF}^2(0,T;H^{1,2})\times \cL_{\sF}^2(0,T;H^{1,2})\times \cL_{\sF}^1(0,T;L^{1})$, and $(u',\psi'+\sigma^*Du',f'(u'))\in\cL_{\sF}^2(0,t;H^{1,2})\times \cL^2_{\sF}(0,t;H^{1,2})\times \cL^1_{\sF}(0,t;L^{1})$  for any $t\in(0,T)$. If there exist $\widetilde L>0$ and $g\in \cL^2(0,T;L^2)$ such that  a.e. in $\Omega\times[0,T]$, $u'(\omega,t)\geq g(\omega,t)$ a.e. in $\bR^d$,
   $$
   \left\langle f(\omega,t,u_t)- f'(\omega,t,u'_t),\, (u_t-u'_t)^+\right\rangle\leq \widetilde L\left\| (u_t-u'_t)^+ \right\|^2 \quad \textrm{and} \quad G(\omega,y)\leq G'(\omega,y) ,$$
then we have a.s.
  \begin{equation}\label{eq-prop-comprn}
  u\leq u' \quad \text{ a.e. in }\bR^d,\ \ \forall t\in[0,T].
  \end{equation}
\end{prop}

\begin{proof}
We put $(\bar{u},\bar{\psi})=(u-u',\psi-\psi') $. 
Since $b$, $\sigma$ and $\varrho$ as well as their first-order derivatives are bounded and $\bar u \in H^{1,2}$, the integration by parts yields a constant $L_2 < \infty$ s.t.
\[
	| \langle \bar u^+_t, (b^*_t+\varrho^*\sigma^*) D\bar u_t \rangle | \leq L_2 \langle \bar u^+_t,\bar u^+_t \rangle, \quad \forall t \in [0,T].
\]
Let us denote the entries of the diffusion matrix $\sigma_s$ by $\sigma^{jr}_s$ and the entries of the vector $\bar \psi_s$ by $\bar \psi_s^r$. Then integration by parts gives:
\begin{align*}
          \sigma^{jr}_s\partial_{y_j}\bar{\psi}^r_s=
     \partial_{y_j}\left(   \sigma^{jr}_s\bar{\psi}^r_s \right)-\partial_{y_j}\sigma^{jr}_s\bar{\psi}^r_s \quad \text{and} \quad
\sigma_s^{ir}\sigma^{jr}_s\partial_{y_iy_j}\bar{u}_s
=\partial_{y_i}\left(\sigma_s^{ir}\sigma^{jr}_s\partial_{y_j}\bar{u}_s\right)
 -\partial_{y_i}\left(\sigma_s^{ir}\sigma^{jr}_s\right)\partial_{y_j}\bar{u}_s.
\end{align*}
Hence, It\^o's formula yields a constant $L_3 < \infty$ s.t. (applying the summation convention):
  \begin{align}
    & \, E \|\bar{u}^+_t\|^2- E \|\bar{u}^+_{\tau}\|^2\nonumber  \\
    \leq&\,
     E\bigg[
            \int_t^{\tau}\!\!\!2\left\langle
            \bar{u}^+_s,\,
            -\frac{1}{2}\partial_{y_i}\left(\sigma^{ir}_s\sigma^{jr}_s\right)\partial_{y_j}\bar{u}^+_s
            -\partial_{y_j}\sigma^{jr}_s\bar{\psi}^r_s
            +\varrho ^r_s\bar{\psi}^r_s+L_3 \bar{u}^+_s
              \right\rangle\,ds
        \nonumber\\
        &\quad
            +\!\int_t^{\tau}\!\!\left\langle
            \partial_{y_j}\bar{u}^+_s,
            \,\sigma^{ir}_s\sigma^{jr}_s\partial_{y_i}\bar{u}^+_s
            -2\sigma^{jr}_s
            \left(\bar{\psi}^r_s+\sigma^{ir}_s\partial_{y_i}\bar{u}^+_s \right)
                \right\rangle\,ds
                -\!\int_t^{\tau}\!\!\!\|\bar{\psi}_s1_{\{u> u'\}}\|^2\,ds
            \bigg]
     \nonumber\\
     =&\,
     E\bigg[
            \int_t^{\tau}\!\!2\Big\langle
            \bar{u}^+_s,\,
            \frac{1}{4}\partial_{y_iy_j}\left(\sigma^{ir}_s\sigma^{jr}_s\right)
            \bar{u}^+_s
            +
            \left(\varrho ^r_s-\partial_{y_j}\sigma^{jr}_s\right)
            \left(\bar{\psi}^r_s+\sigma^{ir}_s\partial_{y_i}\bar{u}^+_s\right)
            +\frac{1}{2}\partial_{y_i}\left(\varrho ^r_s\sigma^{ir}_s-\sigma^{ir}_s\partial_{y_j}\sigma^{jr}_s\right)
            \bar{u}^+_s
        \nonumber\\
        &\quad
        +L_3 \bar{u}^+_s
              \Big\rangle\,ds
                -\!\int_t^{\tau}\!\!\!\left\|\bar{\psi}_s1_{\{u> u'\}}+\sigma^*_s D\bar{u}^+_s\right\|^2\,ds
            \bigg]
     \nonumber\\
     \leq&\,
     E\bigg[
            C\int_t^{\tau}\!\!\left\langle
            \bar{u}^+_s,\,
            \bar{u}^+_s
            +
            \big|\bar{\psi}_s1_{\{u> u'\}}+\sigma^*_s D\bar{u}^+_s\big|
              \right\rangle\,ds
                -\!\int_t^{\tau}\!\!\left\|\bar{\psi}_s1_{\{u> u'\}}+\sigma^*_s D\bar{u}^+_s\right\|^2\,ds
            \bigg]\ \ \ \ (\textrm{by (H.1),\ (H.2))}
     \nonumber\\
     \leq&\,
      E\bigg[
            C\int_t^{\tau}\!\left\|
            \bar{u}^+_s \right\|^2\,ds
                -\frac{1}{2}\!\int_t^{\tau}\!\!\!\left\|\bar{\psi}_s1_{\{u> u'\}}+\sigma^*_s D\bar{u}^+_s\right\|^2\,ds
            \bigg],\quad 0\leq t<\tau<T.
            \nonumber
  \end{align}
Thus, an application of {Gronwall's} inequality leads to:
  \begin{align*}
    E \|\bar{u}^+_t\|^2
    \leq C  E \|\bar{u}^+_{\tau}\|^2,
  \end{align*}
  with the constant $C$ independent of $t$ and $\tau$. Since $\bar u\leq |u|+|g|$, an application of Fatou's lemma yields \eqref{eq-prop-comprn} because
  \begin{align*}
  E\int_{0}^T \left\|\bar u^+_t\right\|^2 \, dt
  \leq CT\limsup_{\tau\uparrow T}E\left\| \bar u_{\tau}^+ \right\|^2
  \leq CTE\int_{\bR^d}\left[ \limsup_{\tau\uparrow T} \left|\bar u^+_{\tau}(y)\right|^2  \right]dy\,=0.
  \end{align*}
\end{proof}
}
{
\begin{rmk}
In Proposition \ref{prop-comprn}, we see that the random field $u'_t$ is allowed to take an infinite and thus singular terminal value. This property is essentially used in the proof of Theorem \ref{thm-minimal} for the uniqueness of solution to BSPDE \eqref{BSPDE-singlr}. 
\end{rmk}
}

\subsection{Truncated BSPDEs} 

A direct computation shows that $(u,\psi)$ solves \eqref{BSPDE-singlr} if and only if
$(v,\zeta){\triangleq}(\theta u,\theta\psi)$ solves the BSPDE
  \begin{equation}\label{BSPDE-singlr-1}
  \left\{\begin{array}{l}
  \begin{aligned}
  -dv_t(y)=\, &\bigg[\textrm{tr}\left(\frac{1}{2}\sigma_t\sigma_t^{*}
  D^2 v_t(y)
        + D \zeta_t\sigma^{*}_t(y)\right)+\tilde{b}^{*}_tD v_t(y)
+\beta_t^{*}\zeta_t(y)+c_tv_t(y)\\
        &\,+\theta(y)F(t,y,\theta^{-1}(y)v_t(y))   \bigg]\, dt
        -\zeta_t(y)\, dW_{t}, \quad (t,y)\in [0,T)\times \bR^d;\\
    v_T(y)=\,& +\infty,  \quad y\in\bR^d
    \end{aligned}
  \end{array}\right.
\end{equation}
with
  \begin{align*}
    \tilde{b}_t^i(y){\triangleq}\,&b^i_t(y) +2q(1+|y|^2)^{-1}\sum_{j=1}^d {(\sigma_t\sigma^{*}_t)}^{ij}(y)y^j, \ \ i=1,\ldots,d,\\
    \beta^r_t(y){\triangleq} \,&2q(1+|y|^2)^{-1}\sum_{j=1}^d\sigma_t^{jr}(y)y^j,\ \ r=1,\ldots,m,\\
    c_t(y) {\triangleq}\,&q(1+|y|^2)^{-1}\bigg(\textrm{tr}(\sigma_t\sigma_t^{*}(y))
    +\sum_{i=1}^d2y^ib^i_t(y)+2(q-1)(1+|y|^2)^{-1}\sum_{i,j=1}^d{(\sigma_t\sigma^{*}_t)}^{ij}(y)y^iy^j\bigg).
  \end{align*}

This BSPDE (\ref{BSPDE-singlr-1}) has a unique solution if the terminal value is finite. More precisely, for $N\in\bN^+$ we put $$\hat{F}(t,y,\phi(y)) {\triangleq} F(t,y,|\phi(y)|),\quad (t,y,\phi)\in \bR_+\times \bR^d\times L^{0}(\bR^d)$$ and consider then the family of BSPDEs: 
\begin{equation}\label{BSPDE-singlr-N}
  \left\{\begin{array}{l}
  \begin{aligned}
  -dv^N_t(y)=\, &\bigg[\textrm{tr}\left(\frac{1}{2}\sigma_t\sigma^{*}_t
  D^2 v^N_t(y)
        + D \zeta^N_t\sigma^{*}_t(y)\right)+\tilde{b}^{*}_tD v^N_t(y)
+\beta_t^{*}\zeta^N_t(y)+c_tv^N_t(y)\\
        &\,+\theta(y)\hat{F}(t,y,\theta^{-1}(y)v^N_t(y))\bigg]\, dt
        -\zeta^N_t(y)\, dW_{t}, \quad (t,y)\in [0,T]\times \bR^d;\\
    v^N_T(y)=\,& N\theta(y),  \quad y\in\bR^d.
    \end{aligned}
  \end{array}\right.
\end{equation}


\begin{prop}\label{prop-N}
Assume (H.1) and (H.2). For each $N\in\bN^+$ and $p\in[2,\infty)$, BSPDE (\ref{BSPDE-singlr-N}) has a unique solution $(v^N,\zeta^N)$ with
  $$(v^N,\zeta^N+\sigma^*Dv^N)\in \left(S^w_{\sF}([0,T];H^{1,2})\cap C^w_{\sF}([0,T]; H^{1,p})\right)\times L^{2}_{\sF}(0,T;H^{1,2}),$$
   s.t. $\theta^{-1} v^N\in \cL^{\infty}_{\sF}(0,T;L^{\infty}(\bR^d))$ and for arbitrary $\varphi\in C_c^{\infty}(\bR^d)$:
  \begin{align}
    \langle \varphi,\, v^N_t\rangle=\,  \langle\varphi,\,N\theta\rangle+
    \int_t^{T}\!&\left\langle\varphi,\,\textup{tr}\left( \frac{1}{2}\sigma_s\sigma^*_s
  D^2 v^N_s
        + D \zeta^N_s\sigma^{*}_s\right)
        +\tilde{b}^{*}_sD v^N_s+c_sv^N_s+\beta^{*}_s\zeta^N_s
        +\theta\hat{F}(s,\theta^{-1}v^N_s)   \right\rangle\, ds
        \nonumber
\\
        &\,
        -\int_t^{T}\!\!\left\langle\varphi,\,\zeta^N_s\right\rangle\, dW_{s}\quad{\text{a.s.},\, \forall \, 0\leq t\leq T.}\nonumber
  \end{align}
\end{prop}
\begin{proof}
To prove existence of a solution, one truncates the quadratic term in $\hat F$ at some level $M$ using a smooth truncation function as in the proof of Proposition 4.1 in \cite{GraeweHorstQui13}. For each $M\in\bN_0$, we know from Theorem \ref{prop-bspde-DTZ} that the resulting BSPDE 
has a unique solution $(v^{N,M},\zeta^{N,M})$ with $$(v^{N,M},\zeta^{N,M}+\sigma^*Dv^{N,M})\in \left(S^w_{\sF}([0,T];H^{1,2})\cap C^w_{\sF}([0,T]; H^{1,p})\right)\times L^{2}_{\sF}(0,T;H^{1,2}).$$
Changing the coefficients $(\lambda,\gamma,M)$ in the above BSPDE to $(\Lambda,+\infty,0)$ we get a new equation. For this equation, one readily checks that
\[
	(\hat{v}_t(y),0)\triangleq(\theta(y)\left(N+\Lambda(T-t)\right),0)
\]
is a solution, and the comparison principle stated in Proposition \ref{prop-comprn} yields:
$$
0\leq v^{N,M}_t\leq \hat{v}_t\quad \ \text{a.e. in }\bR^d, \text{ $\forall\,t\in[0,T]$, {a.s.}}
$$
Choosing $M\in\bN^+$ large enough, we see that 
%
$(v^{N,M},\zeta^{N,M})$  is a solution to \eqref{BSPDE-singlr-N}. Uniqueness of solutions follows from a similar arguments.
\end{proof}

\begin{cor}\label{cor-prop-N}
Assume that the coefficients of the BSPDE \eqref{BSPDE-singlr-N} satisfy Conditions (H.1)-(H.2) and denote the solution by $(v^N,\zeta^N)$. Let $(\tilde{\lambda},\tilde{\gamma},\tilde{\eta})$ be another set of coefficients which satisfies the same conditions as $(\lambda,\gamma,\eta)$. Let $G\in L^2(\Omega,\sF_T;H^{1,2})$ and
  $$(\tilde{v},\tilde{\zeta})\in S^w_{\sF}([0,T];H^{1,2})\times L^{2}_{\sF}(0,T;L^2)$$
  with $\theta^{-1} \tilde{v}\in \cL^{\infty}_{\sF}(0,T;L^{\infty}(\bR^d))$
  being a solution to the BSPDE:
    \begin{equation}\label{BSPDE-singlr-bar}
  \left\{\begin{array}{l}
  \begin{aligned}
  -d\tilde{v}_t(y)=\, &\bigg[\textup{tr}\,\Big(\frac{1}{2}\sigma_t\sigma_t^*
  D^2 \tilde{v}_t(y)
        + D \tilde{\zeta}_t\sigma^{*}_t(y)\Big)+\tilde{b}^{*}_tD \tilde{v}_t(y)
+\beta_t^{*}\tilde{\zeta}_t(y)+c_t\tilde{v}_t(y)+ \theta\tilde{\lambda}_t(y)\\
        &\,
        -\int_{\mathcal{Z}}\frac{\theta^{-1}(y) |\tilde{v}_t(y)|^2}{\tilde{\gamma}_t(y,z)+\theta^{-1} |\tilde{v}_t(y)|}\mu(dz)
        - \frac{\theta^{-1}(y)\left|\tilde{v}_t(y)\right|^2}{\tilde{\eta}_t(y)} \bigg]\, dt
        -\tilde{\zeta}_t(y)\, dW_{t};\\
    \tilde{v}_T(y)=\,& G(y),  \quad y\in\bR^d.
    \end{aligned}
  \end{array}\right.
\end{equation}
If $G\geq  N\theta$, $\tilde{\lambda}\geq \lambda$, $\tilde{\gamma}\geq \gamma$ and $\tilde{\eta}\geq \eta$, then a.s.
$$\tilde{v}_t(y)\geq v^N_t(y)\ \ \ \text{a.e. in }\bR^d,\ \ \forall t\in[0,T].$$
{Moreover, the inequality also holds with all ``$\geq$" replaced by ``$\leq$" in above statement}.
\end{cor}

\end{appendix}

%
%
\vspace{4mm}

\bibliographystyle{siam}

\end{document}